\documentclass[letterpaper, 10 pt, conference]{ieeeconf}  
\let\labelindent\relax
\IEEEoverridecommandlockouts                              
\usepackage{cite}
\usepackage{amsmath,amssymb,amsfonts}
\usepackage{graphicx}
\usepackage{textcomp}
\usepackage{xcolor}
\usepackage{todonotes}
\usepackage{siunitx} 
\def\BibTeX{{\rm B\kern-.05em{\sc i\kern-.025em b}\kern-.08em
		T\kern-.1667em\lower.7ex\hbox{E}\kern-.125emX}}	
\usepackage[T1]{fontenc} 
\usepackage{tikz} 

\usepackage{algorithm} 
\usepackage{graphicx}
\usepackage[utf8]{inputenc}
\usepackage[noend]{algpseudocode}
\usepackage{stackengine} 
\newcommand\barbelow[1]{\stackunder[1.2pt]{$#1$}{\rule{.8ex}{.075ex}}}
\usepackage{multirow} 
\usepackage{pgfplots} 
\usepackage{nccmath} 
\usepackage{lipsum} 
\usepackage{subcaption}
\usepackage{enumitem} 
\usepackage{empheq} 
\usepackage{stackengine} 
\setstackEOL{\\} 

\usepackage{tcolorbox} 
\colorlet{yellow1}{yellow!15}
\colorlet{alg_box}{gray!15}

\newtcolorbox{highlightbox}[1]{
    colback=#1,
    boxrule=0pt,
    boxsep=0pt,
    arc=0pt,
    left=0pt,
    right=0pt,
    top=1pt,
    bottom=1pt,
    width=1.05\linewidth, 
    height=1.3cm,  
}
\usepackage{booktabs} 
\usetikzlibrary{arrows, positioning, shapes.geometric}
\usetikzlibrary{fit, backgrounds} 

\tikzstyle{startstop} = [rectangle, rounded corners, 
minimum width=3cm, 
minimum height=1cm,
text centered, 
draw=black, 
text width=3.6cm, 
fill=red!20]

\tikzstyle{io} = [trapezium, 
trapezium stretches=true, 
trapezium left angle=70, 
trapezium right angle=110, 
minimum width=2.8cm, 
minimum height=1.4cm, 
text centered,
text width=4.5cm, 
draw=black, fill=blue!20]

\tikzstyle{process} = [rectangle, 
minimum width=3cm, 
minimum height=1cm, 
text centered, 
text width=3.8cm, 
draw=black, 
fill=orange!20]

\tikzstyle{decision} = [diamond, 
minimum width=2cm, 
minimum height=0.7cm, 
text centered, 
draw=black, 
text width=2cm, 
fill=green!20]

\tikzstyle{arrow} = [thick,->,>=stealth]
\usepackage{pdfpages}
  
\usepackage{defs}

\algrenewcommand\algorithmicindent{0.5em} 
\begin{document}
	
	\title{\LARGE \bf Parallel Domain-Decomposition Algorithms for Complexity Certification of Branch-and-Bound Algorithms for Mixed-Integer Linear and Quadratic Programming} 

\author{Shamisa Shoja*, Daniel Arnström**, Daniel Axehill*
	\thanks{*S. Shoja and D. Axehill are
		with the Division of Automatic Control, Department of
		Electrical Engineering, Linköping University, Sweden. Email: 
		{\tt\small \{shamisa.shoja,  daniel.axehill\}@liu.se}. **D. Arnström is with the Division of Systems and Control, Department of Information Technology, Uppsala University, Sweden. Email: 
		{\tt\small daniel.arnstrom@it.uu.se}.  \protect\newline
  This work was partially supported by the Wallenberg AI, Autonomous
Systems and Software Program (WASP) funded by the Knut and Alice
Wallenberg Foundation. The National Academic Infrastructure for Supercomputing in Sweden (NAISS) is also acknowledged for providing access to the Sigma  HPC cluster.} 
}

\maketitle
\thispagestyle{empty}
\pagestyle{empty}

		\begin{abstract}
		When implementing model predictive control (MPC) for hybrid systems with a linear or a quadratic performance measure, a mixed-integer linear program (MILP) or a mixed-integer quadratic program (MIQP) needs to be solved, respectively, at each sampling instant. Recent work has introduced the possibility to certify the computational complexity of branch-and-bound (B\&B) algorithms when solving MILP and MIQP problems formulated as multi-parametric MILPs (mp-MILPs) and mp-MIQPs. Such a framework allows for computing the worst-case computational complexity of standard B\&B-based MILP and MIQP solvers, quantified by metrics such as the total number of LP/QP iterations and B\&B nodes. 
	These results are highly relevant for real-time hybrid MPC applications. 
	In this paper, we extend this framework by developing parallel, domain-decomposition versions of the previously proposed algorithm, allowing it to scale to larger problem sizes and enable the use of high-performance computing (HPC) resources. 
	Furthermore, to reduce peak memory consumption, we introduce two novel modifications to the existing (serial) complexity certification framework, integrating them into the proposed parallel algorithms. Numerical experiments show that the parallel algorithms significantly reduce computation time while maintaining the correctness of the original framework. 
		\end{abstract}
	
	
		\newtheorem{lemma}{Lemma}
		\newtheorem{corollary}{Corollary}
		\newtheorem{theorem}{Theorem}
		\newtheorem{assumption}{Assumption}
		\newtheorem{remark}{Remark}
		\newtheorem{definition}{Definition}	
		\newtheorem{properties}{Property}				
\section{Introduction}
To determine optimal control inputs in model predictive control (MPC), an optimization problem is solved at each time step. When MPC is applied to hybrid systems that contain both continuous and discrete 
dynamics, the resulting optimization problems, depending on the chosen performance measure, can be formulated as mixed-integer linear programs (MILPs) (for $1$-norm or $\infty$-norm) or mixed-integer quadratic programs (MIQPs) (for $2$-norm)~\cite{dua2002multiparametric,borrelli2005dynamic,borrelli2017predictive}.  
In MPC applications, system states and reference signals are usually assumed to belong to a closed (in many applications polyhedral) set and are considered as parameters in the optimization problem. This allows for the formulation of these problems as multi-parametric MILP (mp-MILP) and mp-MIQP problems~\cite{bemporad2002model}. 
Computing the solutions to these optimization problems parametrically offers the advantage of shifting most of the online computational burden offline~\cite{dua2000algorithm}. However, storing pre-computed solutions might require a significant amount of memory and can become increasingly complex for high-dimensional problems, limiting its use on embedded hardware.
The traditional alternative to solving these optimization problems parametrically offline is to solve these problems online in real-time, which in the MPC application considered in this work means solving MILPs and MIQPs under real-time constraints. In such a setup, it is clearly of high interest to be able to obtain relevant bounds on the worst-case computational complexity. Such guarantees on the worst-case computational complexity for LPs, QPs, MILPs, and MIQPs have been provided in, e.g.,~\cite{zeilinger2011real}, ~\cite{cimini2017exact, arnstrom2021unifying}, ~\cite{shoja2022exact, shoja2023subopt}, and~\cite{axehill2010improved, shoja2022overall,shoja2025unifying}, respectively. 
In particular, the unified certification method presented in~\cite{shoja2025unifying} 
focuses on determining the worst-case computational complexity
for solving MILP and MIQP instances to optimality that originate from a given mp-MILP or mp-MIQP and a corresponding polyhedral parameter set, when employing the standard branch-and-bound (\bnb) algorithm.  
This certification is performed 
through the process of recursively partitioning the parameter space based on the solver state sequence. 
The complexity measures considered in~\cite{shoja2022overall,shoja2022exact,shoja2025unifying} are the number of linear systems of equations that are solved in the subsolver of \bnb (the \textit{iteration} number) and the number of relaxations (nodes) required to be solved in \bnb. 
Nevertheless, this method might require a considerable amount of processing time and memory, particularly as the size of the problem increases.

The primary contribution of this work is proposing complexity certification algorithms that can be executed in parallel, with the aim to distribute both the computational load and memory requirements across a (potentially large) number of workers~\cite{constantinides2009tutorial, bertsekas2015parallel, ralphs2016parallel}. 
In particular, our objective is to introduce parallel versions of our previously presented complexity certification framework, which will open up for the consideration of larger problem sizes and the certification of more challenging problems. This not only improves computation time but also enables the use of high-performance computing (HPC) to further expand the range of problem sizes that can be certified within a reasonable timeframe. 

When certifying larger problem instances in this work, a significant increase in peak memory consumption was observed. To address this, we propose two key modifications to the (serial) complexity certification framework presented in~\cite{shoja2025unifying}. These modifications are then incorporated into the proposed parallel algorithms to reduce peak memory usage.
To summarize, the main contributions of this paper are:  
\begin{itemize}  
    \item Two parallel algorithms for certifying the complexity of (serial) B\&B-based MILP and MIQP solvers.  
    \item Two strategies for reducing peak memory consumption.  
\end{itemize}

\section{Problem formulation}  \label{sec:probform_par}
In this work, we consider mp-MILP problems of the form

\begin{subequations} \label{eq:mpMILP_par}
	\begin{align}
		\min_{x} \quad &  c^T x \label{eq:mpMILP_par1} \\ 
		\mathcal{P}_{\text{mpMILP}}(\theta): \hspace{.2cm} \textrm{s.t.} \quad & Ax \leqslant  b(\theta), 
        \label{eq:mpMILP_par2} \\
		& x_i \in  \{0,1\}, \hspace{.2cm} \forall i \in \mathcal{B}, \label{eq:mpMILP_par3}
	\end{align}
\end{subequations}
where  $x \in \mathbb{R}^{n_c} \times \{0, 1\}^{n_b}$ is the decision variable vector with $n = n_c + n_b$ components, including $n_c$ continuous and $n_b$ binary components, and $\theta$ is the parameter vector  belonging to the polyhedral set $\Theta_0 \subset \mathbb{R}^{n_{\theta}}$.
The objective function in~\eqref{eq:mpMILP_par} is defined by $c \in \mathbb{R}^n$, and the feasible set is characterized by the matrix $A \in \mathbb{R}^{m \times n}$ and the vector $b(\theta)$, which is an affine function of $\theta$,  given by $b(\theta) = b + W \theta$, where $b \in \mathbb{R}^{m}$ and $W \in \mathbb{R}^{m \times n_{\theta}}$. 
Moreover, $\mathcal{B}$ represents the index set of the binary decision variables.
The problem \eqref{eq:mpMILP_par} is non-convex and is known to be $\mathcal{NP}$-hard \cite{wolsey2020integer}. 

A convex relaxation of~\eqref{eq:mpMILP_par} can be derived by relaxing the binary constraints~\eqref{eq:mpMILP_par3} as follows 
\begin{subequations} \label{eq:mpLP_par}
	\begin{align}
		\min_{x} \quad & c^T x \label{eq:mpLP_par1} \\ 
		\mathcal{P}_{\text{mpLP}}(\theta): \hspace{.2cm} \textrm{s.t.} \quad & Ax \leqslant  b(\theta), 
        \label{eq:mpLP_par2} \\
		& 0 \leqslant x_i  \leqslant 1, \hspace{.2cm} \forall i \in \mathcal{B}, \label{eq:mpLP_par3}\\
		& x_i = 0, \hspace{.08cm} \forall i \in \mathcal{B}_0, \hspace{.1cm} x_i = 1, \hspace{.08cm} \forall i \in \mathcal{B}_1, \label{eq:mpLP_par4}
	\end{align}
\end{subequations}
which is in the form of an mp-LP problem. Here, $\mathcal{B}_0 \subseteq \mathcal{B}$ and $\mathcal{B}_1 \subseteq \mathcal{B}$ 
represent the indices of binary variables fixed to $0$ and $1$, respectively, with $\mathcal{B}_0 \cap \mathcal{B}_1 = \emptyset$. 

For a fixed parameter vector $\bar{\theta} \in \Theta_0$, the problem \eqref{eq:mpMILP_par} is simplified to the (non-parametric) MILP problem $\mathcal{P}_{\text{MILP}}(\bar{\theta})$ with $b(\bar{\theta}) = b + W \bar{\theta}$.  
An LP relaxation $\mathcal{P}_{\text{LP}}(\bar{\theta})$ of the MILP problem 
$\mathcal{P}_{\text{MILP}}(\bar{\theta})$ can be obtained from~\eqref{eq:mpLP_par} by fixing $\theta$ to $\bar{\theta}$. 

Additionally, we consider mp-MIQP problems of the form 

\begin{subequations} \label{eq:mpMIQP_par}		
	\begin{align}
		\vspace{.1cm} 
		\min_{x} \quad & \frac{1}{2} x^{T}Hx + f(\theta)^T x 
        \label{eq:mpMIQP_par1} \\		
		\mathcal{P}_{\text{mpMIQP}}(\theta): \hspace{.3cm} \textrm{s.t.} \quad & Ax \leqslant  b(\theta), 
        \label{eq:mpMIQP_par2} \\
		& x_i \in  \{0,1\}, \hspace{.3cm} \forall i \in \mathcal{B}, \label{eq:mpMIQP_par3}
	\end{align}
\end{subequations}
where $H \in \mathbb{S}_{+}^n$ and $f(\theta)$ is an affine function of $\theta$,  given by $f(\theta) = f + f_{\theta} \theta$, with $f \in \mathbb{R}^{n}$ and $f_{\theta} \in \mathbb{R}^{n \times n_{\theta}}$. 
The decision variables, parameter vector, and feasible set are defined similarly to those in~\eqref{eq:mpMILP_par}. 

A convex relaxation $\mathcal{P}_{\text{mpQP}}(\theta)$ of~\eqref{eq:mpMIQP_par} can similarly be derived by relaxing the binary constraints~\eqref{eq:mpMIQP_par3} into constraints~\eqref{eq:mpLP_par3}. Furthermore, for a fixed parameter $\bar{\theta} \in \Theta_0$, the (non-parametric) MIQP problem $\mathcal{P}_{\text{MIQP}}(\bar{\theta})$ is obtained from~\eqref{eq:mpMIQP_par} with $b(\bar{\theta}) = b + W \bar{\theta}$ and $f(\bar{\theta}) = f + f_{\theta} \bar{\theta}$. 

A commonly used technique to compute the solutions of MILPs and MIQPs is branch and bound (\bnb)~\cite{wolsey2020integer}. 
This method explores a binary search tree in search for a possible optimal solution, where each \textit{node} within the tree represents a convex (LP/QP) relaxation. The relaxations can be solved using an LP/QP solver, such as the simplex or active-set method~\cite{nocedal2006numerical}.  
An important feature of \bnb is that the results from the solutions to the relaxations can be used to prove that parts of the tree can be pruned from explicit exploration. 
Such pruning can be done if one of the following conditions is satisfied in a node~\cite{wolsey2020integer}:
\begin{enumerate}[leftmargin=*]
	\item The objective function value of the relaxation (the lower bound), denoted as $\barbelow{J}$, is greater than the objective function value of the best-known integer-feasible solution so far (the upper bound), denoted as $\bar{J}$. In other words, $\barbelow{J} \geq \bar{J}$, which is referred to as the \textit{dominance cut}. The \textit{infeasibility cut} is a special case of this, where $\barbelow{J} = \infty$.  
	\item The solution to the relaxation, denoted $\barbelow{x}$, is integer feasible, which is referred to as the \textit{integer-feasibility cut}. 
\end{enumerate}

If none of these conditions holds in a node, the tree is branched at that node by fixing a relaxed binary variable, 
which generates two new nodes (child nodes to the parent node). In this paper, we denote a \bnb node representing a relaxation 
by $\eta \triangleq \left ( \mathcal{B}_0, \mathcal{B}_1\right )$,  
where $\mathcal{B}_0$ and $\mathcal{B}_1$ are defined in \eqref{eq:mpLP_par4}. The first node, i.e., the \textit{root node}, is thus denoted $\eta_0 =\left (\emptyset, \emptyset \right )$, where all the binary variables have been relaxed.
For a detailed description of the \bnb method, see, e.g.,~\cite{wolsey2020integer}. 

Throughout this paper, $\mathbb{N}_{0}$ denotes the set of nonnegative integers, $\mathbb{N}_{1:N}$ represents the finite set $\{1, \dots, N\}$, and $\{\mathcal{C}^i\}_{i=1}^{N}$ denotes a finite collection $\{\mathcal{C}^1, \dots, \mathcal{C}^N\}$ of $N$ elements. When $N$ is unimportant, we use  $\{\mathcal{C}^i\}_{i}$ instead. 


\section{Serial complexity certification framework}
\label{sec:cert_ser_par} 
This section briefly reviews the (serial) complexity certification framework for standard B\&B algorithms for MILPs and MIQPs presented in~\cite{shoja2025unifying}, referred to as the \textsc{B\&BCert} algorithm.  
The framework takes an mp-MILP~\eqref{eq:mpMILP_par} or mp-MIQP~\eqref{eq:mpMIQP_par} along with the parameter set \( \Theta_0 \) as inputs and partitions \( \Theta_0 \) based on the computational effort required to reach optimality.  
Specifically, \textsc{B\&BCert} iteratively decomposes \( \Theta_0 \) into regions where the parameters produce the same solver state sequence. The complexity measure, denoted by \( \kappa \), quantifies the number of linear systems of equations solved in relaxations (iterations) or the number of B\&B nodes explored. Moreover, as shown in~\cite{shoja2025unifying}, the framework can be extended to include other relevant complexity measures, such as the number of floating-point operations (flops).

For clarity and completeness, a version of the \textsc{B\&BCert} certification algorithm is provided in Algorithm~\ref{alg:cert_ser}. This algorithm maintains two lists: \( \mathcal{S} \), which stores regions that have not yet terminated, and \( \mathcal{F} \), which stores the terminated regions. Each tuple $\left(\Theta, \mathcal{T}, \kappa_{\text{tot}}, \bar{J}\right)$ in $\mathcal{S}$, denoted  \texttt{reg} for short, consists of the following elements:  
\begin{itemize}
    \item \( \Theta \): the corresponding parameter set.  
    \item \( \mathcal{T} = \{\eta_i\}_i \): a sorted list of \textit{pending} nodes, corresponding to the local \bnb tree within the region.  
    \item \( \kappa_{\text{tot}} \): the accumulated complexity measure obtained for the region to reach its current state.  
    \item \( \bar{J}(\theta) \): the best-known integer feasible solution (i.e., the upper bound) for all \( \theta \in \Theta \).  
\end{itemize}
The initial tuple ($\texttt{reg}^0$) passed to Algorithm~\ref{alg:cert_ser} is $(\Theta^0, \mathcal{T}^0, \kappa_{{tot}}^0, \bar{J}^0 )= (\Theta_0, \{(\emptyset, \emptyset)\}, 0, \infty)$, which contains the entire parameter set $\Theta_0$ and the root node. The output of the algorithm is a partition of $\Theta_0$ along with the corresponding complexity measure. 

\begin{algorithm}[H]
	\caption{\textsc{B\&BCert}: (serial) B\&B complexity certification algorithm~\cite{shoja2025unifying}} 
\label{alg:cert_ser}	
\begin{algorithmic}[1] 
	\Require \Longunderstack[l]
	{$ \texttt{reg}^0 = (\Theta^0, \mathcal{T}^0, \kappa_{{tot}}^0, \bar{J}^0 )$}
	\vspace{.03cm} 
	\Ensure 
	Final partition $\mathcal{F}$ 
	\vspace{.03cm} 
	\State {$\mathcal{F} \leftarrow \emptyset$}
	\State {Push $ \texttt{reg}^0$ 
		to $\mathcal{S}$} 
	\label{step:init_S_par} 
	\vspace{.05cm}
	\While {$\mathcal{S} \neq \emptyset$} \label{step:alg_ser_while}
	\State Pop \texttt{reg} = $( \Theta, \mathcal{T}, \kappa_{{tot}}, \bar{J})$ from $\mathcal{S}$ 
	\label{step:alg_ser_pop}
	\If{$\mathcal{T}=\emptyset$} \label{step:alg_ser_T} 
\State{Add \texttt{reg} to $\mathcal{F}$} 
\label{step:alg_ser_pushF}
\Else
\State{Pop new node $\eta$ 
from $\mathcal{T}$} \label{step:alg_ser_pop_node}
\State $\{(\Theta^j, \kappa^j, \barbelow{J}^j)\}_{j=1}^N \leftarrow \textsc{solveCert}(\eta, \Theta)$ \label{step:alg_ser_cert}
\For {$j \in \mathbb{N}_{1:N} $} \label{step:alg_ser_for} 
\State $\texttt{reg}^j \leftarrow ( \Theta^j, \mathcal{T}, \kappa^j + \kappa_{tot}, \bar{J})$ \label{step:alg_ser_regg}
\State $\mathcal{S} \leftarrow \textsc{cutCert}$ $(\texttt{reg}^j, \barbelow{J}^j, \eta, \mathcal{S})$ \label{step:alg_ser_cut} 
\EndFor
\label{step:pop_ser_cert} 
\EndIf
\EndWhile
\State \textbf{Return} $\mathcal{F}$
\end{algorithmic}
\end{algorithm}

Algorithm~\ref{alg:cert_ser} follows these main steps for a selected \texttt{reg} at Step~\ref{step:alg_ser_pop}:
\begin{itemize}[leftmargin=*]
    \item If no pending nodes remain in \( \mathcal{T} \), \texttt{reg} is added to the final partition $\mathcal{F}$ (Step~\ref{step:alg_ser_pushF}).
    \item Otherwise, a node is selected from \( \mathcal{T} \) (Step~\ref{step:alg_ser_pop_node}), and its corresponding (LP/QP) relaxation is certified using the \textsc{solveCert} subroutine (Step~\ref{step:alg_ser_cert}). This subroutine partitions the parameter set \( \Theta \), associating each resulting subregion \( \Theta^j \) with the optimal value fuction \( \barbelow{J}^j(\theta) \) (lower bound) and the required complexity \( \kappa^j \). An example of \textsc{solveCert} is provided in~\cite{arnstrom2021unifying}.
    \item The region is then decomposed into $N$ subregions (Step~\ref{step:alg_ser_regg}). For each subregion \( \Theta^j \), the parameter-dependent \bnb cut conditions (see Section~\ref{sec:probform_par}, Items 1-3) are evaluated using the \textsc{cutCert} procedure. More details on this procedure can be found in~\cite[Algorithm~8]{shoja2025unifying}.
\end{itemize}


Throughout this paper, we assume that the \textsc{solveCert} function satisfies the properties outlined in Assumption~1 of~\cite{shoja2025unifying} (which hold for certain subsolvers based on results for mp-LPs and mp-QPs), summarized below. 

\begin{assumption} \label{ass:assump_par}
The \textsc{solveCert} function partitions a parameter set \( \Theta \) into a finite collection of  subsets \( \{\Theta^j\}_j \). The value function \( \barbelow{J}(\theta) \) is  piecewise affine (PWA) for mp-LPs and piecewise quadratic (PWQ) for mp-QPs. Additionally, the complexity measure \( \kappa(\theta) \) is piecewise constant (PWC). Moreover, the complexity measure obtained from \textsc{solveCert} is assumed to coincide with that of the corresponding online solver for any \( \theta \in \Theta_0 \). 
\end{assumption}


\section{Parallel complexity certification framework}
\label{sec:cert_par}
In \textsc{B\&BCert}, the entire set \( \Theta_0 \) is processed sequentially, iteratively decomposing it into subregions \( \Theta \subseteq \Theta_0 \) while exploring all pending nodes in the B\&B search tree until an optimal solution is found or infeasibility is detected $\forall \theta \in \Theta$. In this section, we extend this approach to enable parallel execution of the certification algorithm for mp-MILPs and mp-MIQPs in~\eqref{eq:mpMILP_par} and~\eqref{eq:mpMIQP_par}, respectively.  
The following subsections introduce two parallel versions of \textsc{B\&BCert} based on domain decomposition, which distributes tasks across available computational resources, referred to as \textit{workers}. The main processor responsible for this distribution is referred to as the \textit{master}~\cite{ralphs2016parallel}. 

\subsection{Static domain decomposition} \label{sec:paral_dom}

An approach to parallelizing the certification framework 
is to divide the parameter set \( \Theta_0 \) into (artificial) subsets, for example, by partitioning it into equally sized boxes. These artificial regions are then distributed among available processors (workers), with each worker independently applying the \textsc{B\&BCert} algorithm to its assigned region.  

The \textsc{B\&BCertStat} procedure, detailed in Algorithm~\ref{alg:cert_par_dom}, implements this static domain-decomposition strategy. The worker pool, denoted by \( \mathcal{W} \), represents the collection of processors executing tasks concurrently.  Starting with the initial tuple \( \texttt{reg}^0 = (\Theta^0, \mathcal{T}^0, \kappa_{{tot}}^0, \bar{J}^0 )=(\Theta_0, \{(\emptyset, \emptyset)\}, 0, \infty) \) (as in Algorithm~\ref{alg:cert_ser}), the algorithm partitions \( \Theta_0 \) into artificial regions $\{\Theta^k\}_{k=1}^{n_p}$ (Step~\ref{step:alg_dom_split}) and distributes these regions across workers. Each worker \( w \in \mathcal{W} \) independently executes \textsc{B\&BCert} on its assigned region (Step~\ref{step:alg_cert_par_dom_wor}). Finally, the results from completed workers (including the partitioned parameter set and corresponding complexity measures) are aggregated into the final list \( \mathcal{F}_s \) (Step~\ref{step:alg_dom_pushF}). 
The algorithm terminates when all workers have completed their assigned tasks.

\begin{algorithm}[htbp] 
\caption{\textsc{B\&BCertStat}: Static parallel domain-decomposition complexity certification algorithm}
\label{alg:cert_par_dom}
\begin{algorithmic}[1]
\Require \Longunderstack[l]
{$\texttt{reg}^0 = (\Theta^0, \mathcal{T}^0, \kappa_{{tot}}^0, \bar{J}^0 )$}
\vspace{.03cm} 
\Ensure Final partition $\mathcal{F}_s$ 
\vspace{.03cm} 
\State $\mathcal{F}_s \leftarrow \emptyset$
\State Initialize worker pool $\mathcal{W}$
\State Partition $\Theta^0$ into $n_p$ parts $\{\Theta^k\}_{k=1}^{n_p}$ \label{step:alg_dom_split}
\ForAll {$k \in \mathbb{N}_{1:n_p}$ \textbf{in parallel}}
\State  $\texttt{reg}^k \leftarrow \left(\Theta^k, \mathcal{T}^0, \kappa_{\text{tot}}^0, \bar{J}^0\right)$ 
\label{step:alg_dom_dist}
\State $\mathcal{F}^k \leftarrow$ \textsc{B\&BCert}$(\texttt{reg}^k)$ \label{step:alg_cert_par_dom_wor} 
\EndFor
\State Append $\mathcal{F}^k$ for all $ k$ to $\mathcal{F}_s$ 
\label{step:alg_dom_pushF}
\State \textbf{Return} $\mathcal{F}_s$ 
\end{algorithmic}
\end{algorithm}

Algorithm~\ref{alg:cert_par_dom} follows the \textit{master-worker} paradigm~\cite{ralphs2016parallel}, where the master partitions the initial parameter set, 
assigns regions to workers,  
and collects the results. 
While conceptually straightforward, this approach does not exploit the problem's inherent structure. It introduces initial artificial regions, potentially leading to redundant computations and reduced overall efficiency. 

\subsection{Dynamic domain decomposition}  
\label{sec:paral_dist}

This section introduces a parallel method that dynamically distributes computational tasks among available resources while leveraging the problem’s solution structure.  
Unlike static decomposition, this approach assigns tasks to workers based on the relaxation’s solution structure, aiming to improve resource utilization efficiency and minimize redundant computations. 

The \textsc{B\&BCertDyn} procedure, detailed in Algorithm~\ref{alg:cert_par}, implements this dynamic domain-decomposition strategy. 
Similar to Algorithm~\ref{alg:cert_ser}, it maintains two lists: $\mathcal{S}_d$ for candidate regions and $\mathcal{F}_d$ for terminated regions. The algorithm takes as input an initial tuple  
$\texttt{reg}^0 = (\Theta^0, \mathcal{T}^0, \kappa_{\text{tot}}^0, \bar{J}^0)$  
and a user-defined parameter $r \in [0,1]$ that determines the distribution ratio. The output is a partition of $\Theta^0$ along with the corresponding complexity measures.  
Upon its initial invocation, the algorithm starts with  
$\texttt{reg}^0 = (\Theta^0, \mathcal{T}^0, \kappa_{\text{tot}}^0, \bar{J}^0) = (\Theta_0, \{(\emptyset, \emptyset)\}, 0, \infty)$ (as in Algorithm~\ref{alg:cert_ser}). The same underlying \textsc{solveCert} and \textsc{cutCert} functions are used in this algorithm as in Algorithm~\ref{alg:cert_ser}. 

\begin{algorithm}[htbp]
\caption{\textsc{B\&BCertDyn}: Dynamic parallel domain-decomposition complexity certification algorithm}
\label{alg:cert_par}
\begin{algorithmic}[1]
\Require $\texttt{reg}^0 = (\Theta^0, \mathcal{T}^0, \kappa_{\text{tot}}^0, \bar{J}^0)$, ratio $r$ 
\Ensure Final partition $\mathcal{F}_d$
\vspace{.03cm}
\State $\mathcal{F}_d \leftarrow \emptyset$ 
\State Push $\texttt{reg}^0$ to $\mathcal{S}_d$
\State Initialize worker pool $\mathcal{W}$
\While{$\mathcal{S}_d \neq \emptyset$} 
\label{step:alg_par_while}
\State Pop $\texttt{reg} = (\Theta, \mathcal{T}, \kappa_{\text{tot}}, \bar{J})$ from $\mathcal{S}_d$ \label{step:alg_par_pop}
\If{$\mathcal{T} = \emptyset$} \label{step:alg_par_T}
\State Add $\texttt{reg}$ to $\mathcal{F}_d$
\Else
\State Pop new node $\eta$ from $\mathcal{T}$
\State $\{(\Theta^j, \barbelow{J}^j, \kappa^j)\}_{j=1}^N \leftarrow \textsc{solveCert}(\eta, \Theta)$
\For{$j \in \mathbb{N}_{1:N}$}
\State $\texttt{reg}^j \leftarrow (\Theta^j, \mathcal{T}, \kappa^j + \kappa_{\text{tot}}, \barbelow{J}^j)$
\State $\mathcal{S}_d \leftarrow \textsc{cutCert}(\texttt{reg}^j, \barbelow{J}^j, \eta, \mathcal{S}_d)$ \label{step:alg_par_cut}
\EndFor
\State $\mathcal{S}_{w} \leftarrow$ Pop $\lceil r \  |\mathcal{S}_d| \rceil$  tuples from $\mathcal{S}_d$ \label{step:alg_par_dist} 
\ForAll{$k \in \mathbb{N}_{1:|\mathcal{S}_{w}|}$ \textbf{in parallel}} \label{step:cert_par_for1}
\State Pop $\texttt{reg}^k$ from $\mathcal{S}_w$ \label{step:pop_reg_dist}
\State \Longunderstack[l]{ $\mathcal{F}^k  \leftarrow$ \textsc{B\&BCertDyn}$(\texttt{reg}^k,r)$}\label{step:cert_par_wor}
\EndFor
\EndIf
\State Append $\mathcal{F}^k$ for all $ k$ to $\mathcal{F}_d$ 
\label{step:alg_par_pushF}
\EndWhile
\State \textbf{Return} $\mathcal{F}_d$
\end{algorithmic}
\end{algorithm}

Algorithm~\ref{alg:cert_par} leverages the fact that each tuple in $\mathcal{S}_d$ is processed independently, allowing the distribution of regions across multiple independent workers, as implemented in  Steps~\ref{step:pop_reg_dist}--\ref{step:cert_par_wor}. 
The distribution ratio $r \in [0,1]$ at Step~\ref{step:alg_par_dist} specifies the proportion of regions in $\mathcal{S}_d$ to be assigned to workers, where $|\mathcal{S}_d|$ denotes the number of tuples in $\mathcal{S}_d$. The value of  \( r \) 
can be dynamically adjusted based on factors such as the \textit{level} of the node in the \bnb tree, such 
that at the upper levels  
of the tree (closer to the root), more tasks (regions) are distributed to other workers, while at lower levels (deeper in the tree), 
fewer regions are distributed.  
The \textit{for} loop at Step~\ref{step:cert_par_for1} distributes the fraction $r$ of these regions to available workers in $\mathcal{W}$, while the remaining regions are retained for evaluation by the master.  

Each worker then executes the \textsc{B\&BCertDyn} function on its assigned tuple $\texttt{reg}^k$, recursively invoking Algorithm~\ref{alg:cert_par} until termination. 
Once a worker completes its tasks, it returns the results to the master, which aggregates them into the final partition $\mathcal{F}_d$ at Step~\ref{step:alg_par_pushF}. The algorithm terminates when all workers have completed their tasks and no unprocessed regions remain in $\mathcal{S}_d$.

Algorithm~\ref{alg:cert_par} follows a \textit{multiple-master-worker} paradigm \cite{ralphs2016parallel}, where each worker can dynamically assume the role of a master, distributing tasks to idling workers (see Fig.~\ref{fig:topology_par}).  By dynamically distributing tasks and concurrently processing regions, Algorithm~\ref{alg:cert_par} improves scalability 
and is well-suited for large-scale certification problems.  

\begin{figure}[htbp]
\begin{tikzpicture}[
    scale=0.8,
    master/.style={rectangle, draw, rounded corners, minimum width=1.2cm, minimum height=0.9cm, fill=purple!20, text centered},
    worker/.style={rectangle, draw, rounded corners, minimum width=1.2cm, minimum height=0.9cm, fill=cyan!20, text centered},
    subworker/.style={rectangle, draw, rounded corners, minimum width=1.2cm, minimum height=.9cm, fill=green!20, text centered},
    level distance=1.5cm,
    sibling distance=3.5cm,
    every node/.style={scale=0.8},
    edge from parent/.style={draw, -latex},
    level 2/.style={sibling distance=2.5cm}
]

\node[master] (Master0) {Master$^0$}
    child {node[worker] (Worker1) {Worker$^1$ \\ (Master$^1$)}
        child {node[subworker] (Worker11) {Worker$^{1,1}$}}
        child {node[subworker] (Worker12) {Worker$^{1,2}$}}
    }
    child {node[worker] (Worker2) {Worker$^2$ \\ (Master$^2$)}
        child {node[subworker] (Worker21) {Worker$^{2,1}$}}
    }
    child {node[worker] (WorkerP) {Worker$^P$}};

\draw[thick, <-] (Master0) -- (Worker1);
\draw[thick, <-] (Master0) -- (Worker2);
\draw[thick, <-] (Master0) -- (WorkerP);

\draw[thick, <-] (Worker1) -- (Worker11);
\draw[thick, <-] (Worker1) -- (Worker12);

\draw[thick, <-] (Worker2) -- (Worker21);

\end{tikzpicture}
\caption{Topology of the proposed parallel Algorithm~\ref{alg:cert_par}.} 
\label{fig:topology_par}
\end{figure}
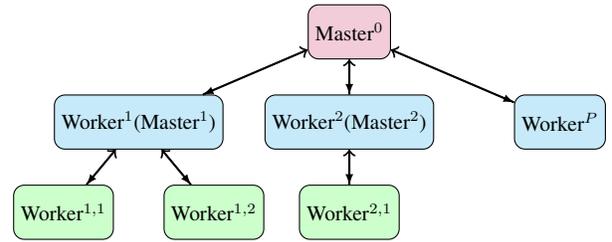

\begin{remark}
At Step~\ref{step:cert_par_wor} of Algorithm~\ref{alg:cert_par}, if the distribution ratio $r$ (input argument) for the worker is set to $0$, then \textsc{B\&BCertDyn} executed by the worker reduces to \textsc{B\&BCert}, as no further task distribution occurs. Consequently, each worker independently executes Algorithm~\ref{alg:cert_ser} on its assigned region. Upon completing its task and certifying the entire region, the worker returns the results directly to the master.  
This special case closely resembles Algorithm~\ref{alg:cert_par_dom}, with the key difference that the parameter set is partitioned only after certifying the root node. This ensures that the final partition coincides with the one generated by the serial algorithm, avoiding artificial subdivisions of the parameter set.
\end{remark}

\subsection{Properties of the parallel algorithms} \label{sec:prop_par}
This section investigates the properties of the proposed parallel algorithms. 
These algorithms spatially parallelize the serial certification framework across the parameter set, while the exploration of nodes in the \bnb search tree remains sequential. The following theorems formally establish that the sequences of \bnb nodes explored in Algorithms~\ref{alg:cert_par_dom} and~\ref{alg:cert_par} are equivalent to those explored by Algorithm~\ref{alg:cert_ser} (the serial counterpart) for all parameters of interest.
\begin{theorem} \label{thr:seq_par_dom} 
Let $\mathbb{B}_s(\theta)$ denote the sequence of nodes explored by Algorithm~\ref{alg:cert_par_dom} (\textsc{B\&BCertStat})  applied to the problem~\eqref{eq:mpMILP_par}/\eqref{eq:mpMIQP_par} for a fixed $\theta \in \Theta_0$ in a terminated region in $\mathcal{F}_s$. Moreover, let $\mathbb{B}(\theta)$ represent the sequence of nodes explored by Algorithm~\ref{alg:cert_ser} (\textsc{B\&BCert}) applied to the same problem for $\theta$ in a terminated region in $\mathcal{F}$. Then,  $\mathbb{B}_s(\theta) = \mathbb{B}(\theta)$, $\forall \theta \in \Theta_0$.
\end{theorem}

\begin{proof}
Both algorithms start with the same initial tuple $\texttt{reg}^0 = (\Theta^0, \mathcal{T}^0, \kappa_{\text{tot}}^0, \bar{J}^0) = (\Theta_0, \{(\emptyset, \emptyset)\}, 0, \infty)$, ensuring identical starting conditions. In Algorithm~\ref{alg:cert_par_dom}, $\Theta_0$ is partitioned into artificial regions $\{\Theta^k\}_{k=1}^{n_p}$, and each such region retains the same initial $\mathcal{T}^0$, $\kappa_{\text{tot}}^0$, and $\bar{J}^0$. By construction, $\cup \Theta^k = \Theta_0$, ensuring that for any $\theta \in \Theta_0$, there is a unique artificial region $\Theta^k$ that contains $\theta$. 

Now, consider an artificial region $\Theta^k$ and a fixed $\theta \in \Theta^k$. Algorithm~\ref{alg:cert_par_dom} processes $\theta$ by applying \textsc{B\&BCert} to $\Theta^k$, while in Algorithm~\ref{alg:cert_ser}, $\theta$ is processed by applying the exact same steps to the entire parameter set $\Theta_0 \ni \theta$. 
In both algorithms, the process is performed in spatially independent regions. In particular, if $ \theta \in \Theta^j$ in Algorithm~\ref{alg:cert_ser}, then the same operation is performed in Algorithm~\ref{alg:cert_par_dom} for $\theta \in \Theta^k \cap \Theta^j$. Given that the initial node list and upper bound are the same in both algorithms, the node exploration process, including calls to \textsc{solveCert} and \textsc{cutCert}, is performed identically for $\theta$ in both cases.   
Consequently, the node exploration sequences remain identical in both algorithms for $\theta$,   
i.e.,  $\mathbb{B}_s(\theta) = \mathbb{B}(\theta)$. 
Since $\theta$ and $k$ are arbitrarily, this holds $\forall \theta \in \Theta_k$, and all $k$, completing the proof.  
\end{proof}

Therefore, although Algorithm~\ref{alg:cert_par_dom} introduces an artificial partitioning of $\Theta_0$, potentially resulting in a finer partitioning in $\mathcal{F}_s$ compared to $\mathcal{F}$, Theorem~\ref{thr:seq_par_dom} ensures that this does not affect the sequence of nodes explored for any $\theta \in \Theta_0$. In other words, the correctness of the node exploration process and its equivalence to the serial algorithm are preserved.

\begin{theorem}	 \label{thr:seq_par}  
Let $\mathbb{B}^i_d(\theta)$ denote the sequence of nodes explored by Algorithm~\ref{alg:cert_par} (\textsc{B\&BCertDyn})  applied to the problem~\eqref{eq:mpMILP_par}/\eqref{eq:mpMIQP_par} for a terminated region $\Theta^i$ 
in $ \mathcal{F}_d $. 
Moreover, let $\mathbb{B}^i(\theta)$ represent the sequence of nodes explored by  Algorithm~\ref{alg:cert_ser} (\textsc{B\&BCert}) applied to the same problem for the same region 
 $\Theta^i$  
in $ \mathcal{F}$. 
Then, $\mathbb{B}_d^i(\theta) = \mathbb{B}^i(\theta)$, $\forall \theta \in \Theta^i$, and all $i$. 
\end{theorem}
 
\begin{proof} \label{thrm:seq_par_proof}
Consider Algorithms~\ref{alg:cert_ser} and~\ref{alg:cert_par}, both initialized with the same initial tuple $\texttt{reg}^0$. 
Now, consider an iteration of Algorithm~\ref{alg:cert_par} starting from Step~\ref{step:alg_par_pop}, where a tuple $\texttt{reg} = (\Theta, \mathcal{T}, \kappa_{\text{tot}}, \bar{J})$
is selected. By inspection, the sequence of operations at Steps~\ref{step:alg_par_T}--\ref{step:alg_par_cut} is identical to Steps~\ref{step:alg_ser_T}--\ref{step:alg_ser_cut} in Algorithm~\ref{alg:cert_ser} when processing $\texttt{reg}$. Algorithm~\ref{alg:cert_par} then follows one of the following cases based on the value of  $r$:  
\begin{itemize}[leftmargin=*]
\item $r = 0$:  
In this case, no subregions are distributed among workers, and  Algorithm~\ref{alg:cert_par} reduces to Algorithm~\ref{alg:cert_ser}.  
Thus, the sequence of nodes explored in $\Theta^i$ remains the same in both algorithms, i.e.,  $\mathbb{B}^i_d(\theta) = \mathbb{B}^i(\theta)$, $ \forall \theta \in \Theta^i$, $\forall i$. 

\item $r \neq 0$:  
In this case, Algorithm~\ref{alg:cert_par} dynamically distributes subregions among workers. Let  
$\texttt{reg}^k = (\Theta^k, \mathcal{T}^k, \kappa_{\text{tot}}^k, \bar{J}^k)$
be selected from $\mathcal{S}_d$ at Step~\ref{step:alg_par_dist} and assigned to a worker $w \in \mathcal{W}$. The worker then applies the same steps (from Step~\ref{step:alg_par_T} to Step~\ref{step:alg_par_cut}) to $\texttt{reg}^k$, just as Algorithm~\ref{alg:cert_ser} would process $\texttt{reg}^k$ at later iterations.  

If the worker $w$ further distributes subregions to other idling workers by recursively calling \textsc{B\&BCertDyn}, the same reasoning applies to each recursion. Each worker processes its assigned subregions independently while following the same node exploration strategy as the serial algorithm. Hence, the sequence of explored nodes remains consistent with the serial Algorithm~\ref{alg:cert_ser}. 
Therefore, for any terminated  $\Theta^i$, Algorithm~\ref{alg:cert_par} explores the same sequence of nodes as Algorithm~\ref{alg:cert_ser}, ensuring  $\mathbb{B}^i_d(\theta) = \mathbb{B}^i(\theta)$, $ \forall \theta \in \Theta^i$, and all $i$.
\end{itemize}  
This completes the proof.  
\end{proof}

An important observation is that, while different parts of the parameter space are generally processed in a different order in the parallel algorithms compared to the serial one, the exploration (and hence certification) of the \bnb process for any fixed parameter \( \theta \in \Theta_0 \) remains unique and follows the same order as in the serial algorithm.  
As a key result, both proposed parallel algorithms yield the same node sequence exploration as the online \bnb algorithm, as summarized in the following corollary.
\begin{corollary} \label{corr:seq_par} 
Let $\mathbb{B}^*(\theta)$ denote the sequence of nodes explored to solve the problem~\eqref{eq:mpMILP_par}/\eqref{eq:mpMIQP_par} for any fixed $\theta \in \Theta_0$ using the online \bnb algorithm (e.g., \cite[Algorithm~1]{shoja2025unifying}).  Moreover, let $\mathbb{B}_s(\theta)$ ($\mathbb{B}_d(\theta)$) represent the sequence of nodes explored by Algorithm~\ref{alg:cert_par_dom}  (Algorithm~\ref{alg:cert_par}) applied to the problem~\eqref{eq:mpMILP_par}/\eqref{eq:mpMIQP_par} for $\theta$ in a terminated region in $\mathcal{F}_s$  ($\mathcal{F}_d$).  
Then,  
$\mathbb{B}^*(\theta) = \mathbb{B}_s(\theta)$ ($= \mathbb{B}_d(\theta)$),  $\forall \theta \in \Theta_0$. 
\end{corollary}

\begin{proof}
The proof directly follows from Theorem~1 in~\cite{shoja2025unifying}, which establishes the pointwise equivalence of the explored nodes in the serial \textsc{B\&BCert}  and the online \bnb algorithms, combined with  Theorem~\ref{thr:seq_par_dom}  (Theorem~\ref{thr:seq_par}) in this paper. 
\end{proof}

As a consequence of Corollary~\ref{corr:seq_par}, the parallel algorithms provide complexity measures that coincide pointwise with the online \bnb method for any fixed \( \theta \in \Theta_0 \), as summarized in the following corollary.

\begin{corollary} \label{corr:exact_cert_parr} 
Let $\kappa^*_{tot}(\theta)$ denote the accumulated complexity measure returned by the online \bnb algorithm (e.g., \cite[Algorithm~1]{shoja2025unifying}) to solve the problem~\eqref{eq:mpMILP_par}/\eqref{eq:mpMIQP_par} for any fixed $\theta \in \Theta_0$. The accumulated complexity measure $\kappa_{tot}(\theta): \Theta_0 \rightarrow \mathbb{N}_0$ returned by Algorithm~\ref{alg:cert_par_dom} (Algorithm~\ref{alg:cert_par}) applied to the problem~\eqref{eq:mpMILP_par}/\eqref{eq:mpMIQP_par} in a terminated region $\Theta^i$ in $\mathcal{F}_s$ ($\mathcal{F}_d$) 
satisfies: (i) $\kappa_{tot} = \kappa^*_{tot}$, $\forall \theta\in \Theta^i$, and all $i$, and (ii) $\kappa_{tot}(\theta)$ is PWC. 
\end{corollary}
\begin{proof} \label{proof_th:exact_cert_parr}
From Assumption~\ref{ass:assump_par}, \textsc{solveCert} correctly certifies the relaxations. Moreover, from Corollary~\ref{corr:seq_par}, Algorithm~\ref{alg:cert_par_dom} (Algorithm~\ref{alg:cert_par})  explores the same sequence of relaxations for a fixed $\theta \in \Theta_0$ as the online \bnb algorithm. As a result, the accumulated complexity measure $\kappa^i_{\text{tot}}$ returned by Algorithm~\ref{alg:cert_par_dom} (Algorithm~\ref{alg:cert_par})  for  $\Theta^i \ni \theta$ is identical to the accumulated complexity measure $\kappa^*_{\text{tot}}$ returned by the online \bnb algorithm for $\theta$. Since $\theta$ and $i$ are arbitrarily, this holds for all $\theta$ and all $i$, completing the proof of (i). The proof of (ii) follows from (i), as the accumulated complexity measure remains constant in any terminated region in $\mathcal{F}_s$ ($\mathcal{F}_d$). 
\end{proof}	
\section{Reducing peak resource consumption}  Algorithm~\ref{alg:cert_ser}, and consequently the parallel Algorithms~\ref{alg:cert_par_dom} and~\ref{alg:cert_par}, can exhibit high memory demands, particularly as the problem dimension increases. A major contributor to this memory intensity is the storage requirement for the number of tuples maintained in the list $\mathcal{S}$ during execution.  In this section, we propose two novel modifications to Algorithm~\ref{alg:cert_ser} to keep the number of stored tuples in \( \mathcal{S} \) relatively low throughout the algorithm's execution, thereby reducing peak memory consumption. These refinements are then incorporated into the proposed parallel algorithms.

To implement these improvements, we introduce a variable called $\textit{state}$ ($\in \{ \texttt{none}, \texttt{"Fin"}, \texttt{"UnFin"}, \texttt{"CC"} \}$) for each region, indicating its current state. The meaning of these states is clarified in the following subsections.

\subsection{Early termination of relaxation's certification} 
\label{subsec:early_term_par}

When evaluating a node within a region in Algorithm~\ref{alg:cert_ser}, the (mp-LP/mp-QP) relaxation 
is certified using \textsc{solveCert} within $\Theta$ at Step~\ref{step:alg_ser_cert}. This process can be computationally intensive, particularly at the root node, where the entire parameter set is considered. Additionally, executing the \textsc{solveCert} procedure to completion often generates a large number of regions, which are stored in $\mathcal{S}$, leading to excessive resource consumption.

To address these challenges, we propose temporarily pausing the \textsc{solveCert} procedure after generating and processing a predetermined number of regions, denoted by \( N^{\text{max}} \).   This allows certification to be paused if it becomes resource-intensive and resumed later. As a result, the number of regions returned by \textsc{solveCert} does not exceed \( N^{\text{max}} \). To implement this approach, each region generated by the \textsc{solveCert} procedure is assigned a \textit{state} as follows:
\begin{itemize}
    \item \texttt{"Fin"} (finished): if the mp-LP/mp-QP relaxation in this region has been fully certified. That is, if it has been solved to optimality or determined to be infeasible or unbounded.
    \item \texttt{"UnFin"} (unfinished): if the \textsc{solveCert} procedure was paused before certification was completed, i.e., before reaching a conclusive result.
\end{itemize}

Algorithm~\ref{alg:cert_ser} is then modified to evaluate cut conditions exclusively in regions with the \texttt{"Fin"} state, i.e., where the relaxation has been completely certified, while unfinished regions are stored in $\mathcal{S}$ for certification at a later stage. For each unfinished region, the information obtained before pausing is preserved, allowing the \textsc{solveCert} procedure to resume efficiently from the paused state rather than restarting from scratch. This explains the input argument "state" passed to \textsc{solveCert} at Step~\ref{step:alg_par_mod_cert}. This warm-starting capability helps mitigate the computational burden associated with certifying complex subproblems.

\subsection{Delayed cut-condition evaluation} 
\label{subsec:eval_cut_ind_par}

In Algorithm~\ref{alg:cert_ser}, parameter-dependent cut conditions are evaluated iteratively over a loop for all resulting (finished) regions at Steps~\ref{step:alg_ser_for}--\ref{step:alg_ser_cut}. This evaluation may further partition each $\Theta^j$ (e.g., when evaluating dominance cuts or selecting branching indices~\cite{shoja2025unifying}), which increases the number of tuples stored in $\mathcal{S}$ and, consequently, the memory usage.

To reduce memory consumption, Algorithm~\ref{alg:cert_ser} can be modified to delay cut-condition evaluations for all newly generated regions. Specifically, instead of immediately evaluating cut conditions at Step~\ref{step:alg_ser_cut} for the node corresponding to finished regions, these regions are stored in $\mathcal{S}$ with their  \textit{state} set to \texttt{"CC"} (cut conditions pending). When a region with the \texttt{"CC"} state is later popped from $\mathcal{S}$, cut conditions are evaluated for its corresponding node, and its state is updated to \texttt{"Fin"} (finished). 
This modification reduces memory usage by deferring region partitioning originally from cut conditions, ensuring that $\mathcal{S}$ contains fewer tuples at any given time.

\subsection{Modified parallel certification algorithm} \label{subsec:mod_alg_par}

To minimize peak memory usage, we integrate the techniques outlined in Sections~\ref{subsec:early_term_par} and~\ref{subsec:eval_cut_ind_par} into the proposed parallel algorithms. In this revised approach, each tuple in $\mathcal{S}$ is expanded to include the lower bound $\barbelow{J}(\theta)$, the region's \textit{state}, and the current node $\eta^c$, which are used during cut-condition evaluation. Consequently, each tuple is now defined as $ \texttt{reg} = (\Theta, \mathcal{T}, \kappa_{\text{tot}}, \bar{J}, \barbelow{J}, \text{state}, \eta^c)$. At the start of the algorithm, the initial tuple is set to $\texttt{reg}^0 = (\Theta_0, \{(\emptyset, \emptyset)\}, 0, \infty, \infty, \texttt{none}, \texttt{none})$. 

The updated version of  Algorithm~\ref{alg:cert_par}, incorporating both modifications, is presented in Algorithm~\ref{alg:cert_par_mod}. In this algorithm, the modifications from Section~\ref{subsec:early_term_par} are implemented at Steps~\ref{step:alg_par_mod_cert} and Steps~\ref{step:alg_par_mod_st1}--\ref{step:alg_par_mod_st2}, while the modifications described in Section~\ref{subsec:eval_cut_ind_par} are implemented at Steps~\ref{step:alg_par_mod_st_cc}--\ref{step:alg_par_mod_cut}.

\begin{algorithm}[htbp]
\caption{\textsc{B\&BCertDynMod}: Modified dynamic parallel domain-decomposition complexity certification algorithm}
\label{alg:cert_par_mod}
\begin{algorithmic}[1]
\Require \Longunderstack[l]{$ \texttt{reg}^0=(\Theta^0, \mathcal{T}^0, \kappa_{{tot}}^0, \bar{J}^0, \barbelow{J}^0, \text{state}^0, \eta^0)$, ratio $r$} 
\vspace{.03cm}
\Ensure Final partition $\mathcal{F}_d$ 
\vspace{.03cm}  
\State {$\mathcal{F}_d \leftarrow \emptyset$}
\State {Push $ \texttt{reg}^0$ to $\mathcal{S}_d$}
\State Initialize worker pool $\mathcal{W}$
\While {$\mathcal{S}_d \neq \emptyset$} 
\State Pop $\texttt{reg} = (\Theta, \mathcal{T}, \kappa_{\text{tot}}, \bar{J}, \barbelow{J}, \text{state}, \eta^c)$ from $\mathcal{S}_d$
\label{step:alg_par_mod_pop}
\If {$\text{state} = \texttt{"CC"}$} \label{step:alg_par_mod_st_cc}
\State $\mathcal{S}_d \leftarrow \textsc{cutCert}$$(\texttt{reg}, \barbelow{J}, \eta^c, \mathcal{S}_d)$ 
\label{step:alg_par_mod_cut} 
\ElsIf {$\text{state} = \texttt{"Fin"}$ \textbf{and} $\mathcal{T} = \emptyset$}
\State Add $\texttt{reg}$ 
to $\mathcal{F}_d$
\label{step:alg_par_mod_pushF}
\Else
\State Pop new node $\eta$ from $\mathcal{T}$ 
\State \Longunderstack[l]{$\{(\Theta^j,\barbelow{J}^j,\kappa^j, \text{state}^j)\}_{j=1}^N\leftarrow\textsc{solveCert}(\eta,\Theta,\text{state},$ \\$N^{\text{max}})$}\label{step:alg_par_mod_cert}
\For {$j \in \{1, \dots, N\}$}
\If {$\text{state}^j = \texttt{"Fin"}$}
\label{step:alg_par_mod_st1}
\State $\texttt{reg}^j \leftarrow$ $(\Theta^j, \mathcal{T}^j, \kappa_{\text{tot}} + \kappa^j, \bar{J}^j, \barbelow{J}^j, \texttt{"CC"}, \eta)$ \label{step:alg_parr_mod_reg1}
\Else
\State $\texttt{reg}^j \leftarrow$ $(\Theta^j, \mathcal{T}^j, \kappa_{\text{tot}} + \kappa^j, \bar{J}^j, \barbelow{J}^j, \texttt{"unFin"}, \eta)$ \label{step:alg_par_mod_st2}
\EndIf
\State Push $\texttt{reg}^j$ to $\mathcal{S}_d$
\label{step:alg_par_mod_st_reg}
\EndFor
\State $\mathcal{S}_{w} \leftarrow$ Pop $\lceil r \  |\mathcal{S}_d| \rceil$  tuples from $\mathcal{S}_d$ \label{step:alg_par_mod_dist}
\ForAll {$k \in \mathbb{N}_{1:| \mathcal{S}_{w}|}$ \textbf{in parallel}} 
\label{step:cert_par_mod_dist_for1}
\State Pop $\texttt{reg}^k$ from $\mathcal{S}_w$ \label{step:pop_reg_dist_mod} 
\State \Longunderstack[l]{ $\mathcal{F}^k  \leftarrow$ \textsc{B\&BCertDynMod}$(\texttt{reg}^k,r)$}\label{step:cert_par_mod_wor} 
\EndFor
\EndIf
\State Append $\mathcal{F}^k$ for all $ k$ to $\mathcal{F}_s$ 
\EndWhile
\State \textbf{Return}  $\mathcal{F}_d$
\end{algorithmic}
\end{algorithm}

To summarize Algorithm~\ref{alg:cert_par_mod}, for a selected $\texttt{reg}$  at Step~\ref{step:alg_par_mod_pop}, one of the following three scenarios occurs based on its state: 

\begin{enumerate} [label=\roman*.]
\item \textit{Evaluating cut conditions}: The cut conditions are evaluated within the region using \textsc{cutCert}, and the node is either cut or branched.  The results are then stored in $\mathcal{S}_d$ (Step~\ref{step:alg_par_mod_cut}).
\item \textit{Terminating the region}: If no pending nodes remain within the region, it is terminated and added to $\mathcal{F}_d$ 
(Step~\ref{step:alg_par_mod_pushF}).
\item \textit{Certifying the relaxation and distributing regions}: The relaxation is certified within the region using \textsc{solveCert} (Step~\ref{step:alg_par_mod_cert}), and the results are stored in $\mathcal{S}_d$ (Steps~\ref{step:alg_par_mod_st1}--\ref{step:alg_par_mod_st_reg}). 
Next, a fraction $r$ 
of regions stored in $\mathcal{S}_d$ is selected 
for distribution among available workers (Step~\ref{step:alg_par_mod_dist}). Each worker then processes its assigned region using the modified  \textsc{B\&BCertDynMod} algorithm (Step~\ref{step:cert_par_mod_wor}).
\end{enumerate}

The updated version of Algorithm~\ref{alg:cert_par_dom} can be obtained by first incorporating these techniques into Algorithm~\ref{alg:cert_ser}, analogously to Algorithm~\ref{alg:cert_par_mod}, resulting in a modified serial algorithm. Consequently, Algorithm~\ref{alg:cert_par_dom} incorporates these strategies by replacing the \textsc{B\&BCert} function at Step~\ref{step:alg_cert_par_dom_wor} with the modified serial algorithm. 

The properties of the modified parallel algorithms follow directly from Theorems~\ref{thr:seq_par_dom}--\ref{thr:seq_par} and Corollaries~\ref{corr:seq_par}--\ref{corr:exact_cert_parr}, with the corresponding modified serial and parallel algorithms substituted accordingly. Furthermore, it is noted that the introduced delays do not alter the sequence of explored nodes.

\section{Numerical experiments}
In this section, the proposed parallel algorithms were applied to randomly generated MILPs and MIQPs in the form of~\eqref{eq:mpMILP_par} and~\eqref{eq:mpMIQP_par} and an MIQP originating from an MPC application.
The \textsc{solveCert} procedure utilizes the certification algorithm described in~\cite{arnstrom2021unifying}.  
The numerical experiments were implemented in Julia (version 1.10.5), with computations for higher-dimensional problems conducted on two compute nodes, each equipped with 32 cores and 96 GB of RAM, using resources from the National Supercomputer Centre (NSC)~\cite{NSC}.

\subsection{Random examples} We first apply the proposed parallel algorithms to randomly generated MILPs and MIQPs in the form of~\eqref{eq:mpMILP_par} and~\eqref{eq:mpMIQP_par}, respectively. The coefficients of these problems were generated as:
$\bar{H}~\sim~\mathcal{N}(0,1)$, $H~=~\bar{H} \bar{H}^T$, $f~\sim~\mathcal{N}(0,1)$, $f_{\theta}~\sim~\mathcal{N}(0,1)$, 
$c~\sim~\mathcal{N}(0,1)$,
$A~\sim~\mathcal{N}(0,1)$, $b~ \sim~\mathcal{U}([0,2])$, and $W~\sim~\mathcal{N}(0,1)$,  
where \( \mathcal{N}(\mu, \sigma) \) denotes a normal distribution with mean \( \mu \) and standard deviation \( \sigma \), and \( \mathcal{U}([l, u]) \) represents a uniform distribution over the interval \( [l, u] \).
The parameter set was defined as $\Theta_0 = \{\theta \in \mathbb{R}^{n_{\theta}} | \hspace{.05cm} |\theta_i| \leq 0.5, \forall i \}$. For Algorithms~\ref{alg:cert_par} and~\ref{alg:cert_par_mod}, different values of the ratio $r$ were selected from $\{1, 0.8, 0.4, 0\}$. Note that the case where $r = 0$ corresponds to the serial algorithm. In the experiments, the number of continuous decision variables was set to $n_c = n_b$, the number of constraints to $m = n + 8$ (where $n=n_c+n_b$), and the number of parameters to $n_{\theta} = \left\lceil n_b / 4 \right\rceil$.  For example, the largest problem had dimensions $n_b = 30$, $n = 60$, $m = 68$, and $n_{\theta} = 8$. For each value of $n_b$, $25$ random problems were generated and solved.  Furthermore, for Algorithm~\ref{alg:cert_par_dom}, the number of artificial regions was set to $n_p = n_{\theta}^{\left\lceil n / 4 \right\rceil - 1}$.  

The proposed parallel algorithms were validated by comparing their results to those of the corresponding serial algorithm. A key finding is that, for both MILP and MIQP problem families, the worst-case complexity measures, including the worst-case accumulated number of iterations ($\kappa^I$) and \bnb\ nodes ($\kappa^N$), were identical to those obtained by the serial algorithm across all experiments,  confirming the correctness of the proposed algorithms.  Fig.~\ref{fig:rand_res_uni} illustrates the average worst-case (wc) complexity measures obtained by applying Algorithms~\ref{alg:cert_par_dom} and~\ref{alg:cert_par} to randomly generated mp-MILPs, as a function of $n_b$ for different problem sizes. 
The trend in the figure confirms that larger problem sizes require significantly more iterations and nodes due to the combinatorial nature of the problem, highlighting the importance of parallelization strategies for certifying computationally demanding cases.

\begin{figure}[htbp]  
\centerline{		\includegraphics[scale=0.38]{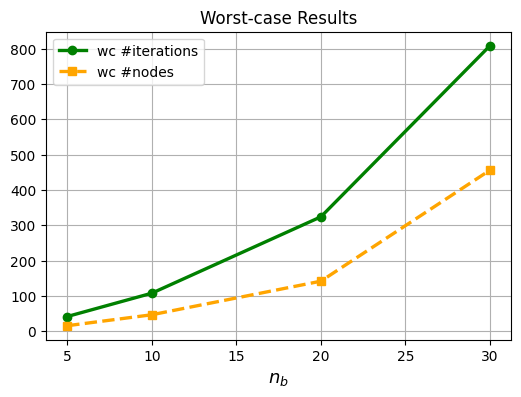}
}	
\caption{Average worst-case number of iterations and B\&B nodes as a function of $n_b$ for random experiments.}  
\label{fig:rand_res_uni}
\end{figure} 

To illustrate the impact of parallel algorithms on computation time, Fig.~\ref{fig:geomean_par} presents the speed-up factor, defined as $\frac{t^s}{t^p}$, where $t^s$ and $t^p$ denote the CPU execution times of the serial and parallel certification algorithms, respectively. These results are based on experiments conducted on a cluster at NSC for random MILPs with $n_b = 30$,  presented as a function of the number of employed workers. For reference, the ideal linear speed-up is also shown.

\begin{figure}[htbp]  
\centerline{		\includegraphics[scale=0.28]{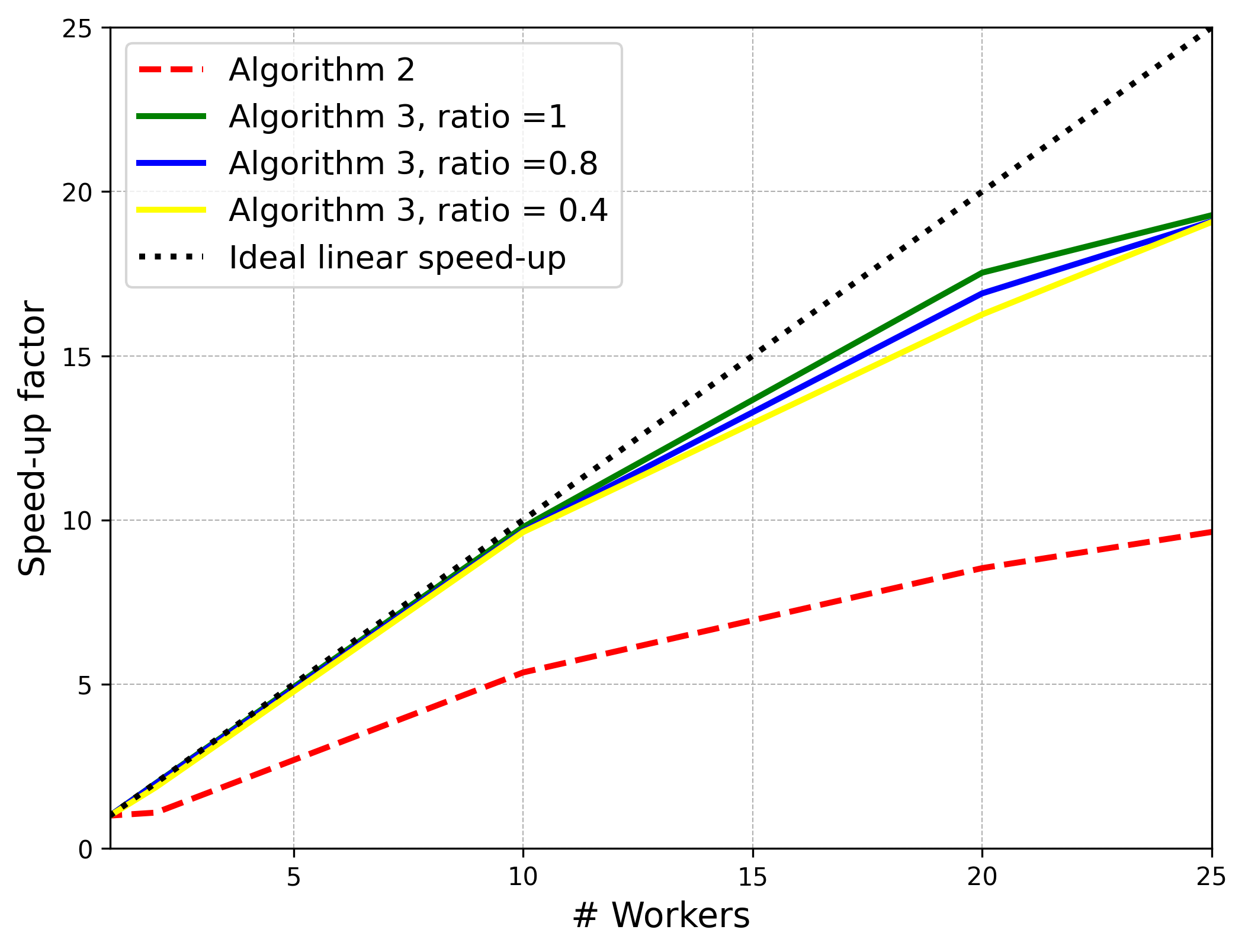}}	
\caption{Speed-up factor of Algorithm~\ref{alg:cert_par_dom} (dashed red line) alongside the speed-up factor of Algorithm~\ref{alg:cert_par} for three different ratios (solid lines), compared to the ideal linear speed-up (dotted black line), as a function of the number of workers for random experiments.}  
\label{fig:geomean_par}
\end{figure}  

From Fig.~\ref{fig:geomean_par}, it is clear that the parallel algorithms significantly accelerate computation. For a larger number of workers, however, the speed-up deviates from the ideal trend, highlighting the emergence of communication overhead. In particular, the speed-up factor closely follows the ideal linear trend for fewer workers when using Algorithm~\ref{alg:cert_par} (and its variant, Algorithm~\ref{alg:cert_par_mod}). Furthermore, larger distribution ratios can lead to more consistent speed-up.  
Algorithm~\ref{alg:cert_par_dom}, on the other hand, provides a moderate improvement over Algorithm~\ref{alg:cert_ser} and achieves a lower speed-up compared to Algorithm~\ref{alg:cert_par}. This is primarily due to the extra regions generated from the initial artificial partitioning, which increase the computational workload (see Fig.~\ref{fig:part_par}~(b) and~(d)).  
These results emphasize the importance of dynamic workload distribution to fully leverage the problem structure and improve performance.  

To illustrate the reduction in peak resource consumption achieved by Algorithm~\ref{alg:cert_par_mod}, Fig.~\ref{fig:LengthS_par} shows the number of tuples stored in $\mathcal{S}_d$ (i.e., $| \mathcal{S}_d |$) in the master over the course of Algorithms~\ref{alg:cert_par} and~\ref{alg:cert_par_mod}. In this experiment, the \textsc{solveCert} subroutine was paused after generating a maximum of \( N^{\text{max}} = 100 \) regions. 
Algorithm~\ref{alg:cert_par_mod} generally requires more outer iterations, since in certain iterations, only the cut conditions are evaluated (see Steps~\ref{step:alg_par_mod_st_cc}--\ref{step:alg_par_mod_cut}). In this experiment, the master in Algorithm~\ref{alg:cert_par_mod} performed approximately 100 more outer iterations than Algorithm~\ref{alg:cert_par}, during which \( |\mathcal{S}_d| = 0 \) for Algorithm~\ref{alg:cert_par}. 
While the total execution time for both algorithms was approximately the same, Algorithm~\ref{alg:cert_par_mod} exhibited lower memory consumption. This is reflected in the maximum number of tuples stored in $\mathcal{S}_d$, 
which is around 100 for Algorithm~\ref{alg:cert_par_mod}, compared to approximately 1000 for Algorithm~\ref{alg:cert_par}.  
Notably, the peaks in Fig.~\ref{fig:LengthS_par} correspond to outer iterations in which a relaxation was certified.

\begin{figure}[htbp]  
\centerline{	\includegraphics[scale=0.28]{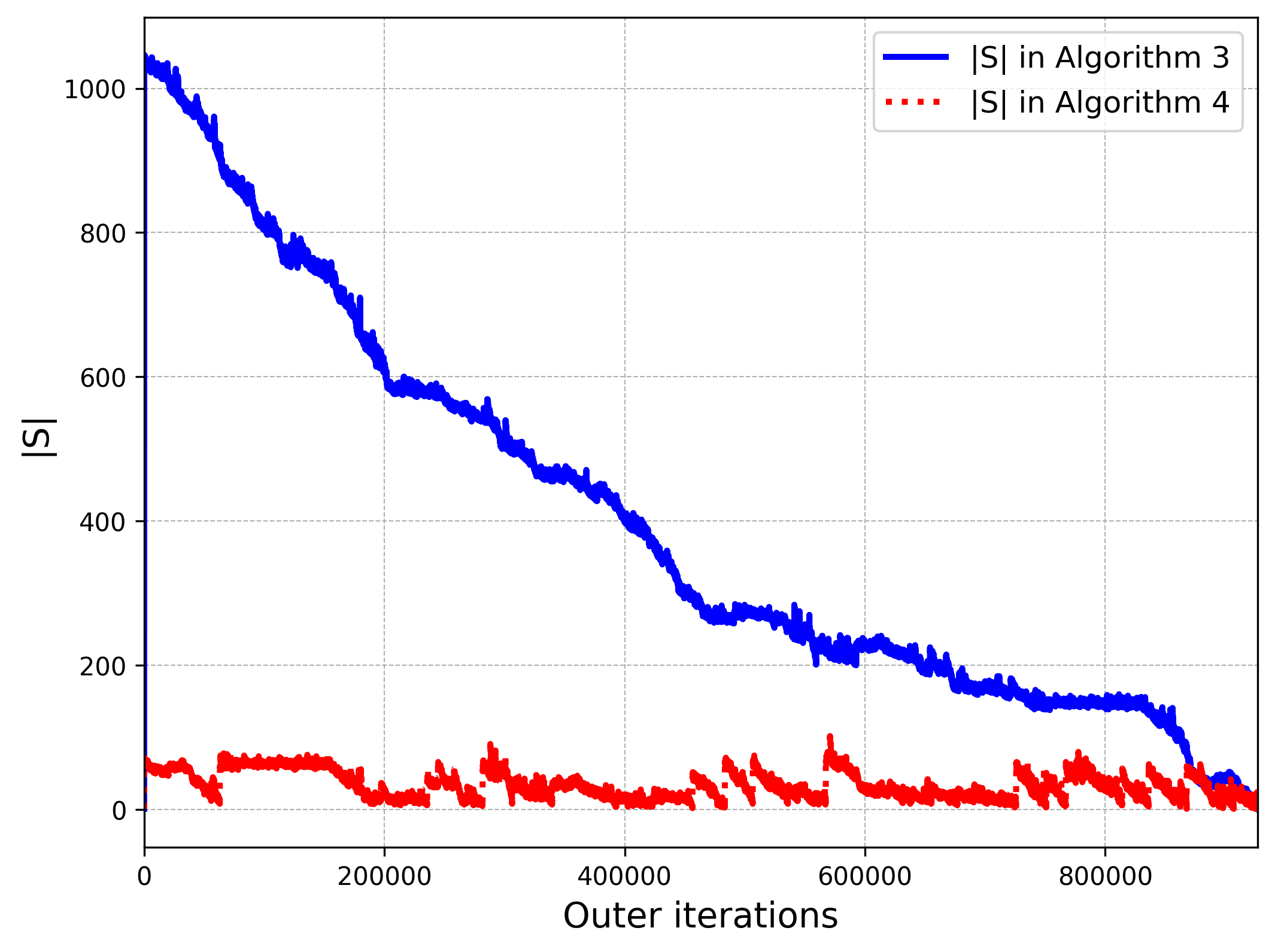}}	
\caption{The number of tuples stored in $\mathcal{S}_d$ in the master during  Algorithms~\ref{alg:cert_par} (solid blue line) and~\ref{alg:cert_par_mod} (dotted red line).} 
\label{fig:LengthS_par}
\end{figure} 

Finally, to provide visual insights into the results of Algorithms~\ref{alg:cert_par_dom} and~\ref{alg:cert_par}, Fig.~\ref{fig:part_par} presents a two-dimensional slice of the resulting partition for two randomly generated problems, with $n_b = 10$ in (a) and (b), and $n_b = 20$ in (c) and (d).  
Specifically, Fig.~\ref{fig:part_par}~(a) and (c) show the final regions generated using  Algorithm~\ref{alg:cert_par}, resulting in 36 and 2510 regions, respectively, identical to the final partitioning obtained using  Algorithm~\ref{alg:cert_ser}. 
Similarly, Fig.~\ref{fig:part_par}~(b) and (d) illustrate the partitions obtained using Algorithm~\ref{alg:cert_par_dom}, yielding 69 and 2696 regions, respectively.  
The additional partitions in Fig.~\ref{fig:part_par}~(b) and (d) arise from the initial artificial partitioning introduced in Algorithm~\ref{alg:cert_par_dom}.  

\begin{figure}[htbp]  
\begin{subfigure}{0.47\columnwidth}		
\centerline{
\includegraphics[scale=0.6]{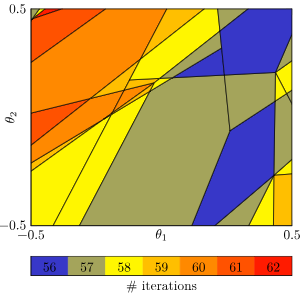}}	\caption{} 
\end{subfigure}
\hfill
\begin{subfigure}{0.47\columnwidth}	
\centerline{
\includegraphics[scale=0.6]{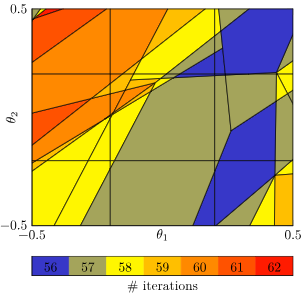}}	
\caption{} 
\end{subfigure}
\hfill
\begin{subfigure}{0.47\columnwidth}		
\centerline{
\includegraphics[scale=0.6]{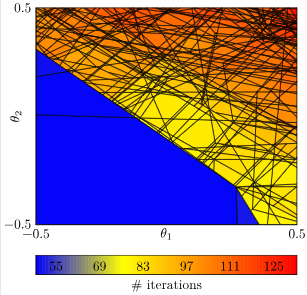}}	\caption{} 
\end{subfigure}
\hfill
\begin{subfigure}{0.47\columnwidth}	
\centerline{
\includegraphics[scale=0.6]{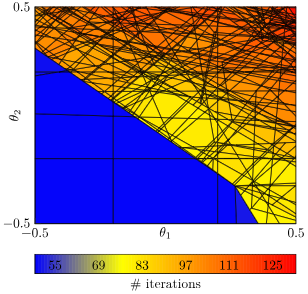}}	
\caption{} 
\end{subfigure}
\caption{Resulting parameter space partitions for two randomly generated examples, determined using Algorithm~\ref{alg:cert_par} in (a) and (c), and Algorithm~\ref{alg:cert_par_dom} in (b) and (d). Results are shown for \(n_b = 10\) in (a) and (b), and \(n_b = 20\) in (c) and (d).}  
\label{fig:part_par}
\end{figure}
\subsection{MPC application}
We now apply the proposed methods to a hybrid MPC problem involving a linearized inverted pendulum on a cart, constrained by a wall, as used in, e.g.,~\cite{marcucci2020warm,shoja2025unifying} (see Fig.\ref{fig:pend_par}). The objective is to stabilize the pendulum at the origin ($z_1 = 0$) while keeping it upright  ($z_2 = 0$). The control inputs consist of a force applied to the cart ($u_1$) and a contact force from wall interaction ($u_2$), introducing binary variables into the model. The state and input constraints follow those in\cite{marcucci2020warm,shoja2025unifying}.
The hybrid MPC controller employs a quadratic performance measure with a prediction/control horizon of $N$. The system’s initial state vector serves as the parameter vector $\theta$ in the multi-parametric setting, with the parameter set determined by the state constraints. The resulting mp-MIQP, formulated following~\cite{bemporad2002model}, includes $n = 4N$ decision variables, of which $n_b = 2N$ are binary. The weight matrices and dynamics are consistent with those in~\cite{marcucci2020warm}. For further details, see~\cite{marcucci2020warm}.  


\begin{figure}[htbp]  
\centerline{	\includegraphics[scale=0.23]{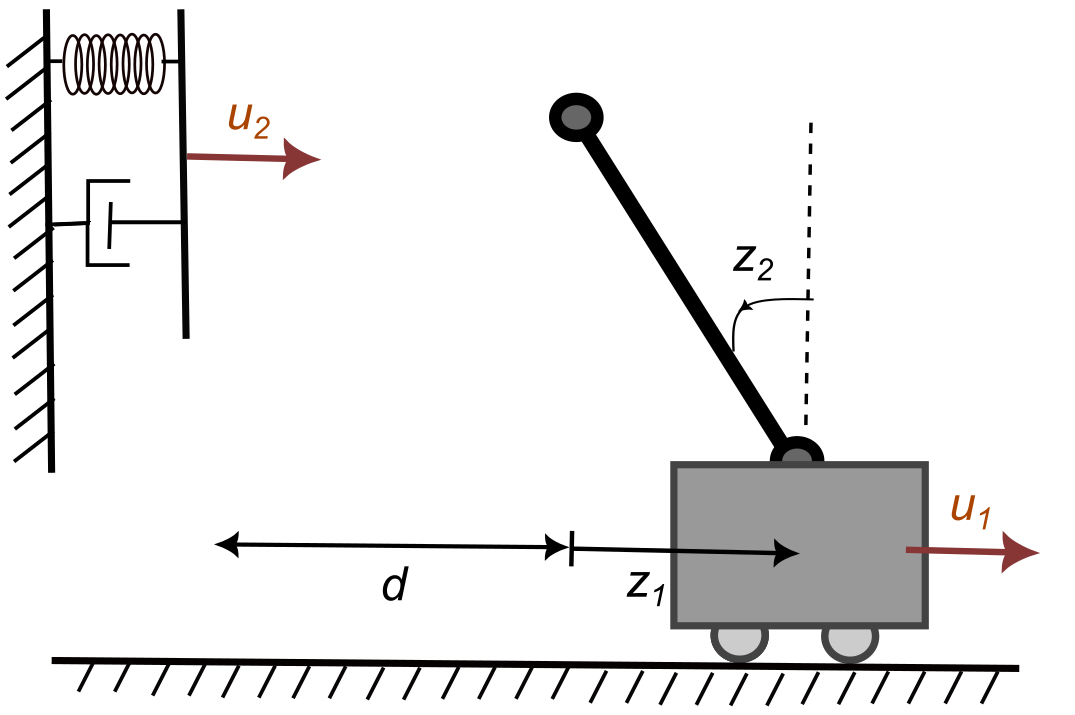}} 
\caption{Regulating the inverted pendulum on a cart with contact forces.} 
\label{fig:pend_par}
\end{figure}

Table~\ref{tab:res_pend_par} summarizes the complexity certification results for the (serial) certification Algorithm~\ref{alg:cert_ser} and the parallel certification Algorithms~\ref{alg:cert_par_dom} and~\ref{alg:cert_par} executed using $10$ workers, as a function of the prediction horizon $N$. The table reports the number of constraints ($m$), the worst-case computational complexity in terms of iterations ($\kappa^I_{\text{wc}}$) and B\&B nodes ($\kappa^N_{\text{wc}}$), along with the CPU time required for certification (in seconds). The results highlight the impact of parallelization in reducing computation time. Specifically, as the problem size increases (i.e., as $N$ grows), the parallel algorithms outperform the serial one in terms of certification time, while maintaining correct results.

\begin{table}[]
\caption{The number of constraints ($m$) in the resulting mp-MIQPs for the inverted pendulum example, along with worst-case complexity results and the computation time of the Julia implementation for Algorithm~\ref{alg:cert_ser} ($t^{\text{ser}}$), Algorithm~\ref{alg:cert_par_dom} ($t^{\text{stat}}$), and Algorithm~\ref{alg:cert_par} ($t^{\text{dyn}}$) across different prediction horizons ($N$).} 
\label{tab:res_pend_par}
\centering
\scalebox{1}{
\begin{tabular}{cc|cc|ccc}  
\multicolumn{2}{c|}{prob. dim.}  & \multicolumn{2}{c|}{wc. complexity}            & \multicolumn{3}{c}{cert. time [sec]}              
\\ \hline
$N$ & $m$ & $\kappa^I_{\text{wc}}$ & $\kappa^N_{\text{wc}}$ & $t^{\text{ser}}$ & $t^{\text{stat}}$ & $t^{\text{dyn}}$ \\ \hline
1   & 19  & 9   &   4   & 1.4  & 0.93    & 0.48    \\ 
2   & 38  &  22   & 8 & 6.2  & 3.2        & 1.9     \\ 
3  &  57   &    76     &  47  &  44.6     &     20.3    &    11.1   \\
4   & 76  &   134  &  68
&   122     &   46.4      &  30.1                      \\
5     &  95  & 198  & 76   &   635 &  201.3  &   119.4   \\    
6    & 114 & 250  & 78   &   1311.6 &   326.9   &   212.7                     
\end{tabular}
}
\end{table}

\section{Conclusion}
This paper presents parallel versions of the complexity certification framework for B\&B-based MILP and MIQP solvers applied to the family of mp-MILPs and mp-MIQPs, extending its applicability to larger and more challenging problem instances. The proposed algorithms employ both static and dynamic domain decomposition to distribute computational workloads across available processing resources.
To address the challenge of peak memory consumption in computationally demanding problems, two complementary strategies are introduced and integrated into the parallel algorithms. These enhancements improve the scalability of the method, even for more memory-intensive problems. The parallel algorithms enable the use of HPC resources, facilitating the certification of larger and more complex problem instances.
Numerical experiments demonstrate significant reductions in computation time while preserving the correctness of the certification results. Furthermore, the introduced memory-handling strategies are shown to significantly reduce peak memory consumption. Future work will focus on further refining the implementation, including parallelizing the  \textsc{solveCert} subroutine to further distribute the computational workload. 
		
		\bibliography{IEEEabrv,references} 

\end{document}